%% file: main.tex
\documentclass[11pt]{article}
\input{preamble.tex}

\DeclareMathOperator{\Unif}{Unif}
\DeclareMathOperator{\normal}{normal}
\DeclareMathOperator{\distinct}{distinct}
\DeclareMathOperator{\DB}{DB}
\newcommand{\DBnormal}{\DB_{\normal}}
\newcommand{\DBdistinct}{\DB_{\distinct}}
\newcommand{\pnormal}{p_{\normal}}
\newcommand{\pdistinct}{p_{\distinct}}

\newenvironment{myabstract}%
{\list{}{\listparindent 1.5em
        \itemindent    \listparindent
        \leftmargin    1cm
        \rightmargin   1cm
        \parsep        0pt}%
    \item\relax}%
{\endlist}

\newenvironment{mycover}%
{\list{}{\listparindent 0pt
        \itemindent    \listparindent
        \leftmargin    1cm
        \rightmargin   1cm
        \parsep        0pt}%
    \raggedright
    \item\relax}%
{\endlist}

\newcommand{\myemail}[1]{\,$\cdot$\, {\small #1}}
\newcommand{\myaff}[1]{\,$\cdot$\, {\small #1}\par\medskip}

\begin{document}

\begin{mycover}
{\huge\bfseries 2-Coloring Cycles in One Round \par}

\bigskip
\bigskip

\textbf{Maxime Flin}
\myemail{maxime.flin@aalto.fi} \myaff{Aalto University, Finland}

\textbf{Alesya Raevskaya}
\myemail{alesya.raevskaya@aalto.fi} \myaff{Aalto University, Finland}

\textbf{Ronja Stimpert}
\myemail{ronja.stimpert@aalto.fi} \myaff{Aalto University, Finland}

\textbf{Jukka Suomela}
\myemail{jukka.suomela@aalto.fi} \myaff{Aalto University, Finland}

\textbf{Qingxin Yang}
\myemail{qing.yang@helsinki.fi} \myaff{Aalto University, Finland}

\bigskip
\end{mycover}

\begin{myabstract}
\fakeparagraph{Abstract.}
We show that there is a one-round randomized distributed algorithm that can 2-color cycles such that the expected fraction of monochromatic edges is less than 0.24118. We also show that a one-round algorithm cannot achieve a fraction less than 0.23879. Before this work, the best upper and lower bounds were 0.25 and 0.2. Our proof was largely discovered and developed by large language models, and both the upper and lower bounds have been formalized in Lean 4.
\end{myabstract}

\thispagestyle{empty}
\setcounter{page}{0}
\clearpage

\section{Introduction}
We study a problem that appears so simple and elementary that it is difficult to believe it remains an open question in 2026. Yet, this seems to be the case.

We have a directed cycle formed by $n$ identical computers, and we seek a randomized distributed algorithm that $2$-colors the cycle in $1$ round.
Obviously, we cannot find a proper $2$-coloring (it does not even exist in odd cycles), so the goal is to \emph{minimize the fraction of monochromatic edges}. Put otherwise, we want to find a \emph{cut that is as large as possible}.

Let $p(A)$ denote the expected fraction of monochromatic edges when we use a $1$-round algorithm $A$, and let $p^*$ be the infimum of these values over all algorithms (see \cref{sec:def} for the detailed definition). What can we say about $p^*$? Prior work \cite{hirvonen-rybicki-etal-2017-large-cuts-with-local} together with a trivial lower bound implies
$
    1/5 \le p^* \le 1/4
$.
But what is the precise value? In this work we develop new techniques for answering this question and show that
\begin{equation}\label{eq:pstar-new}
    0.23879 \le p^* < 0.24118.
\end{equation}
This improves the ratio between lower and upper bounds from $0.80$ to approx.\ $0.99$.

\subsection{Motivation and broader context: distributed quantum advantage}

While this problem is of independent interest---see e.g.\ \cite{hirvonen-rybicki-etal-2017-large-cuts-with-local,shearer-1992-a-note-on-bipartite-subgraphs-of-triangle}---our original motivation for studying it arises from the pressing open question of understanding distributed quantum advantage. One of the biggest open questions in this area boils down to whether 3-coloring cycles (or any local symmetry-breaking problem) admits a distributed quantum advantage, see e.g.\ \cite{akbari-coiteux-roy-etal-2025-online-locality-meets,balliu-brandt-etal-2025-distributed-quantum-advantage,d-amore-2025-on-the-limits-of-distributed-quantum,le-gall-rosmanis-2022-non-trivial-lower-bound-for-3,gavoille-kosowski-markiewicz-2009-what-can-be-observed}. Currently, we cannot even exclude the possibility that quantum-LOCAL can 3-color cycles in 1 round, with probability 1.

Hence, as a first step towards understanding such questions, we started to explore the limitations of problems related to coloring in quantum-LOCAL in 1 round. Clearly, quantum-LOCAL cannot properly 2-color in 1 round, but could it maybe find a 2-coloring with fewer monochromatic edges?
At this point we finally realized: if we design a quantum algorithm, what do we compare it with? For example, if we achieved $p = 0.245$ with quantum algorithms, is it better than the classical limits? Apparently the classical limits for this seemingly simple problem were wide open before this work.

\subsection{Coworkers: large language models and proof assistants}

Almost all research work reported here was done with large language models, primarily with GPT-5.2 in Codex. They not only handled the technical details of the proofs, but also were essential in the discovery of the high-level proof technique that we showcase here.

As it is hard to verify nontrivial LLM-generated mathematical proofs, we also had LLMs formalize both the upper bound and the lower bound in Lean 4 proof assistant, so that the proof can be automatically verified with a computer; see \cite{2-coloring-1-round-lean-nonanon} for the full Lean 4 code.

Hence we believe this work also serves as a tangible demonstration that modern generative AI tools have reached a point where (1)~research-level problems in the theory of distributed computing can be tackled with LLMs, and (2)~automatic formalization of such results is also feasible in practice.

\section{Precise problem definition}\label{sec:def}

As we consider randomized $1$-round algorithms in directed cycles, a distributed algorithm is, in essence, a function that maps the local random bits of $3$ consecutive nodes into the final output. To make the setting as general as possible, we do not limit the number of random bits. This can be conveniently modeled so that each node picks a random real number uniformly from $[0,1]$; note that such a real number can be interpreted as an infinite sequence of random bits. Hence, an algorithm is simply a function
\begin{equation}
    f\colon [0,1]^3\to\{0,1\},
\end{equation}
where $f(a,b,c)$ is the output of a node whose predecessor has random value $a$, its own random value is $b$, and its successor has random value $c$.
To make probabilities well-defined, we require that $f$ is measurable.

\begin{remark*}
One can fully define such an algorithm by specifying which coordinates $(x,y,z)$ in the unit cube satisfy $f(x,y,z) = 1$. Especially if it produces an output that is not scattered over the domain it makes sense to visualize this region as exemplified in \cref{fig:one-param,fig:upper-bound}.
\end{remark*}

With this definition of a $1$-round algorithm, let us now define what is the expected fraction $p(f)$ of monochromatic edges when we apply it to an $n$-cycle, for any $n \ge 4$. Everything is symmetric, so we can consider any edge $e = (u,v)$ in the cycle and calculate the probability $p(f)$ that this edge is monochromatic. This only depends on four random numbers, $A,B,C,D$, that correspond to the random numbers selected by the predecessor of $u$, node $u$, node $v$, and the successor of $v$, in this order. The output of $u$ is $f(A,B,C)$ and the output of $v$ is $f(B,C,D)$. Hence
\begin{equation}\label{eq:p(f)}
    p(f) = \Pr\bigl[f(A,B,C)=f(B,C,D)\bigr],
\end{equation}
where $A,B,C,D$ are i.i.d.\ $\Unif[0,1]$.

\begin{remark*}
Notably, \eqref{eq:p(f)} does not depend on the length of the cycle $n$, as long as $n \ge 4$ (and hence the predecessor of $u$ is independent of the successor of $v$). In particular, the definition of $p(f)$ does not change if we assumed, say, that $n > 1000$ or that $n$ is even.
\end{remark*}

Now $p(f)$ indicates how good a given algorithm $f$ is. As we are interested in the \emph{best possible} algorithm, let us finally define
$
    p^* = \inf_f p(f)
$.
What can we say about the precise value of~$p^*$?

\section{Simple observations and prior work}

To gain some intuition on this problem, let us consider three simple algorithms and one simple lower bound.
First, we can simply produce a random coloring with 
$
    f_1(a,b,c) = 1
$
if $b \ge 1/2$,
and it is easy to see that this results in
$
    p(f_1) = 1/2
$.

Second, we can try to output an independent set. If we succeed in doing this, we will never have monochromatic edges of the form $(1,1)$, only monochromatic edges of the form $(0,0)$, and if our independent set is large, such edges are rare. Here is a simple algorithm that finds an independent set, using the idea that local maxima join:
$
    f_2(a,b,c) = 1 \text{ if } b > a \text{ and } b > c
$.
Analyzing this is an easy exercise:
$
    p(f_2) = 1/3
$.

Third, we can use the idea familiar from e.g.\ \cite{hirvonen-rybicki-etal-2017-large-cuts-with-local}: find a random cut and use 1 round to ``improve'' it. The obvious rule for improvement is this: if I am on the same side as my predecessor and successor, I will change sides. This rule $f_3$ results in 
$
    p(f_3) = 1/4
$,
which is the tightest upper bound we have seen in prior work.

Let us finally give a simple lower bound result. Let $f$ be any algorithm. We can construct a $5$-cycle with $5$ i.i.d.\ random numbers, and apply $f$ to each consecutive $3$-tuple of these numbers. There has to be at least one monochromatic edge in the output, and we obtain
$
    p(f) \ge 1/5
$.
Note that we cannot do the same trick with $3$-cycles, as it would violate the assumption that $A,B,C,D$ are independent. We could do the same trick with $7$-cycles, but the bound would be weaker.

\begin{remark*}
Even if the algorithm $f$ is only promised to work in, say, even cycles of length at least one million, we can still do the ``what if'' experiment and see what \emph{would} happen if we broke the rules and applied it in a $5$-cycle. Function $f$ cannot tell the difference between $5$-cycles and valid inputs, and it has to produce \emph{some} labeling of the $5$-cycle. But if $p(f) < 1/5$, it would have to output on average a $2$-coloring that does not exist in $5$-cycles.
\end{remark*}

We draw attention to the fact that all these bounds are seemingly ad-hoc and unrelated to each other, not following any systematic design strategy, and they do not present any clear path towards proving tight bounds.

\section{Contribution and technical overview}

Our main contribution is this (see \cref{sec:def} for definitions):
\begin{theorem}
    $0.23879 \le p^* < 0.24118$.
\end{theorem}
To prove this, we introduce a systematic approach for deriving close-to-tight upper and lower bounds on $p^*$, by studying two variants of $3$-dimensional De Bruijn graphs.

\subsection{Normal vs.\ distinct De Bruijn graphs}

Let $[n] = \{0,1,\dotsc,n-1\}$.
For all $n \ge 2$, we define the usual $3$-dimensional De Bruijn graph of $n$ symbols, in brief the \emph{normal De Bruijn graph} $\DBnormal(n) = (V_{\normal}(n), E_{\normal}(n))$, as follows:
\begin{align*}
    V_{\normal}(n) &= \bigl\{ (a,b,c) : a,b,c \in [n] \bigr\}, \\
    E_{\normal}(n) &= \bigl\{ \bigl( (a,b,c), (b,c,d) \bigr) : a,b,c,d \in [n] \bigr\}.
\end{align*}
And then for all $n \ge 4$ we define a subgraph of $\DBnormal$ that we call the \emph{distinct De Bruijn graph} $\DBdistinct(n) = (V_{\distinct}(n), E_{\distinct}(n))$, as follows:
\begin{align*}
    V_{\distinct}(n) &= \bigl\{ (a,b,c) : a,b,c \in [n],\ a,b,c\text{ are distinct} \bigr\}, \\
    E_{\distinct}(n) &= \bigl\{ \bigl( (a,b,c), (b,c,d) \bigr) : a,b,c,d \in [n],\ a,b,c,d\text{ are distinct} \bigr\}.
\end{align*}

\begin{figure}[p]
    \centering
    \includegraphics[page=4,width=.6\linewidth]{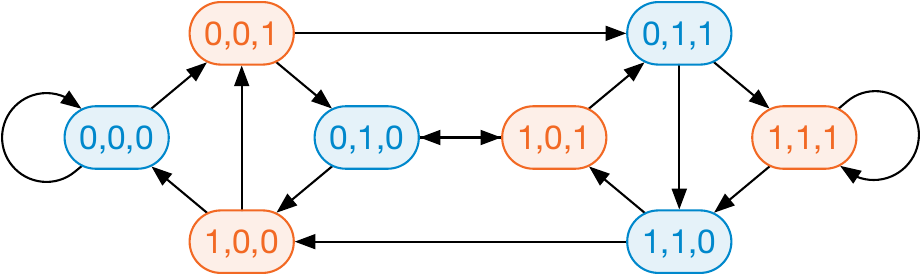}
    \caption{The normal De Bruijn graph $\DBnormal(n)$ for $n=2$, together with a $2$-coloring (blue--orange) where only $4$ out of $16$ edges are monochromatic. It is also easy to see that there is no $2$-coloring with fewer monochromatic edges (consider two self-loops and two triangles). Hence $p^* \le \pnormal(2) = 1/4$.}\label{fig:db-normal-2}
\end{figure}
    
\begin{figure}[p]
    \centering
    \includegraphics[page=5,width=.65\linewidth]{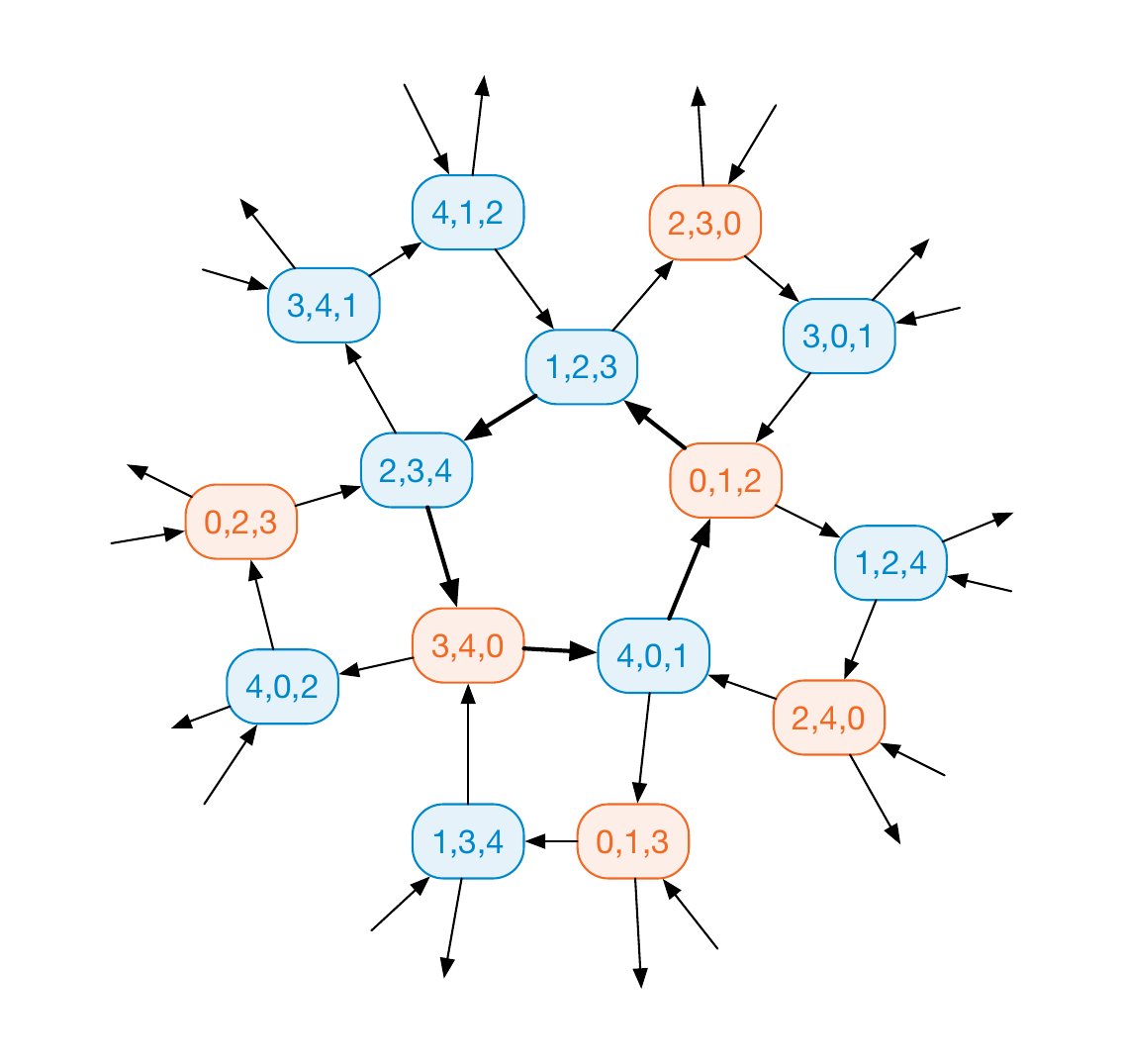}
    \caption{Selected parts of the distinct De Bruijn graph $\DBdistinct(n)$ for $n=5$. The graph has a total of $60$ nodes and $120$ edges. The edges can be partitioned into $24$ cycles of length $5$: for each edge $\bigl((a,b,c),(b,c,d)\bigr)$ there is a unique element $e \in [5] \setminus \{a,b,c,d\}$, and we assign this edge to the $5$-cycle formed by $(a,b,c)$, $(b,c,d)$, $(c,d,e)$, $(d,e,a)$, and $(e,a,b)$. In each such $5$-cycle, at least one edge is monochromatic; therefore, there is no $2$-coloring with fewer than $24$ monochromatic edges. On the other hand, exactly $24$ monochromatic edges can be achieved, e.g., by coloring nodes with labels $(0,y,z)$ and $(x,y,0)$ orange and all other nodes blue, as shown in the figure. Hence $p^* \ge \pdistinct(5) = 1/5$.}\label{fig:db-distinct-5}
\end{figure}

\subsection{Sandwiching}

We write $\pnormal(n)$ for the smallest possible fraction of monochromatic edges in any $2$-coloring of $\DBnormal(n)$, and similarly $\pdistinct(n)$ for the smallest possible fraction of monochromatic edges in $\DBdistinct(n)$; see \cref{fig:db-normal-2,fig:db-distinct-5} for examples. Our main observation is that (1)~we can sandwich $p^*$ always between $\pnormal(n)$ and $\pdistinct(n)$, using any value of $n$, and (2)~this technique is complete in the sense that $\pnormal(n)$ and $\pdistinct(n)$ tend to $p^*$ as $n$ increases. Hence we can replace the seemingly open-ended question of studying all kinds of $1$-round algorithms over an uncountably infinite domain with a very down-to-earth question of studying cuts in De Bruijn graphs.

\begin{lemma}
    \label{lem:upper}
    For all $n \ge 2$, we have $p^* \le \pnormal(n)$.
\end{lemma}
\begin{proof}
    Assume that $h\colon V_{\normal} \to \{0,1\}$ is a $2$-coloring of $\DBnormal(n)$ with fraction $\pnormal(n)$ monochromatic edges. Then define the discretization function $g\colon [0,1] \to [n]$ in the natural way, $g(1) = n-1$ and otherwise $g(x) = \lfloor xn \rfloor$. Define a function
    \[
        f(a,b,c) = h\bigl(g(a), g(b), g(c)\bigr).
    \]
    Now $f$ is a $1$-round $2$-coloring algorithm with $p(f) = \pnormal(n)$. To see this, note that the following processes are equivalent:
    \begin{enumerate}
        \item Pick a uniformly random edge $e = \bigl((A,B,C), (B,C,D)\bigr)$ from $\DBnormal(n)$ and see whether it is monochromatic.
        \item Pick a uniformly random tuple $(A',B',C',D')$ and discretize it to $A=g(A'),B=g(B'),C=g(C'),D=g(D')$ and see whether $h(A,B,C) = h(B,C,D)$. \qedhere
    \end{enumerate}
\end{proof}

\begin{lemma}
    \label{lem:lower}
    For all $n \ge 4$, we have $p^* \ge \pdistinct(n)$.
\end{lemma}
\begin{proof}
    Fix any function $f$; we need to show that $p(f) \ge \pdistinct(n)$.
    Let us recall the definition of $p(f)$: we pick a random tuple $(A,B,C,D)$ and see if $f(A,B,C) = f(B,C,D)$. Let us do the same experiment in a slightly roundabout manner. First, sample $X_1,\dotsc,X_n$ independently from $\Unif[0,1]$. Then pick uniformly at random $4$ distinct indexes $A',B',C',D' \in [n]$. Now $X_{A'},X_{B'},X_{C'},X_{D'}$ is still a uniformly random sample, and we have $f(X_{A'},X_{B'},X_{C'}) = f(X_{B'},X_{C'},X_{D'})$ with probability $p(f)$.

    We now do the following thought experiment. After fixing the random values $X_1,\dotsc,X_n$, define a coloring $h\colon V_{\distinct} \to \{0,1\}\colon (a,b,c) \mapsto f(X_a, X_b, X_c)$. This is a coloring of $\DBdistinct$, and hence the fraction of monochromatic edges must be at least $\pdistinct(n)$. So a random edge $e = \bigl((A',B',C'),(B',C',D')\bigr)$ is monochromatic with probability at least $\pdistinct(n)$. But now $e$ is monochromatic iff $f(X_{A'},X_{B'},X_{C'}) = f(X_{B'},X_{C'},X_{D'})$.
    So conditioned on \emph{any} choice of $X_1,\dotsc,X_n$, with at least probability $\pdistinct(n)$ we must have $f(X_{A'},X_{B'},X_{C'}) = f(X_{B'},X_{C'},X_{D'})$. Therefore $p(f) \ge \pdistinct(n)$ and also $p^* \ge \pdistinct(n)$.
\end{proof}

\begin{lemma}
    We have $p^* = \lim_{n \to \infty} \pnormal(n) = \lim_{n \to \infty} \pdistinct(n)$.
\end{lemma}
\begin{proof}
    We argue that $\lim_{n \to \infty} \bigl(\pnormal(n) - \pdistinct(n)\bigr) = 0$, and the claim follows by \cref{lem:upper,lem:lower} and the sandwich theorem.
    
    Clearly, $|V_{\normal}| = n^3$ and $|E_{\normal}| = n^4$. On the other hand, $V_{\normal} \setminus V_{\distinct}$ contains the vertices $(a,a,a)$, $(a,b,b)$, $(b, a, b)$ and $(b, b, a)$ for $a \neq b \in [n]$, which are fewer than $4n^2$ in total. There are two kinds of edges in $E_{\normal} \setminus E_{\distinct}$: those incident to $V_{\normal} \setminus V_{\distinct}$, and for each $(a, b, c) \in V_{\distinct}$ the three edges to $(b, c, a)$, $(b, c, b)$ and $(b, c, c)$. Overall $E_{\normal} \setminus E_{\distinct}$ contains at most $7n^3$ edges.
    Consider a coloring of $\DBdistinct(n)$ whose fraction of monochromatic edges is $\pdistinct(n)$ and extend it to a coloring of $\DBnormal(n)$ arbitrarily. It has $\pdistinct(n)|E_{\distinct}|$ monochromatic edges on $\DBdistinct(n)$ and all edges of $E_{\normal} \setminus E_{\distinct}$ could be monochromatic. By minimality of $\pnormal(n)$, we have
    \[
    \pnormal(n) \leq \frac{\pdistinct(n)|E_{\distinct}| + |E_{\normal} \setminus E_{\distinct}|}{|E_{\normal}|}
    \leq \pdistinct(n) + \frac{7}{n} \ . \qedhere
    \]
\end{proof}

\subsection{Lower bound}

Starting with small values of $n$, one can observe $\pdistinct(n) = 0.2$ for $5 \le n \le 8$ and $\pdistinct(n) > 0.2$ for $n \ge 9$; hence we could prove a \emph{nontrivial} lower bound by analyzing cuts in $\DBdistinct(9)$. However, for our main result we will use much larger values of $n$.
\begin{lemma}
We have $p^* \ge \pdistinct(10^6) \ge 0.23879$.
\end{lemma}
\begin{proof}[Proof idea]
We use a semidefinite-program (SDP) relaxation \cite{vandenberghe-boyd-1996-semidefinite-programming} of the max-cut problem \cite{goemans-williamson-1995-improved-approximation}, with correlation triangle inequalities \cite{deza-laurent-1997-geometry-of-cuts-and-metrics}. Graph $\DBdistinct(n)$ is highly symmetric, and by exploiting symmetries \cite{gatermann-parrilo-2004-symmetry-groups-semidefinite}, the size of the SDP is bounded by a constant independent of $n$. We use an SDP solver to find a feasible dual solution \cite{peyrl-parrilo-2008-computing-sum-of-squares}, which can be turned into a certificate for $\pdistinct(10^6) \ge 0.23879$. We refer to the formalized proof \cite{2-coloring-1-round-lean-nonanon} for more details.
\end{proof}

\subsection{Upper bound}

We can check with exhaustive enumeration that $\pnormal(2) = 0.25$ and $\pnormal(4) \approx 0.2422$; hence cuts in $\DBnormal(4)$ would already give a significant improvement over prior work. We could prove the upper bound by constructing $\DBnormal(n)$ for a moderately large value of $n$ and using some heuristic approach to find a large cut. However, this would result in an algorithm that does not have any clear structure or short description.

We steered heuristic search towards ``simple'' colorings of $\DBnormal(n)$; here a particularly useful restriction was to require that the resulting function $f(a,b,c)$ is monotone in each coordinate $a,b,c$. Then we visualized the function, and made some educated guesses on the general shape.

The first milestone was the discovery of an algorithm family $f^\tau$ that has a very simple piecewise-linear structure, illustrated in \cref{fig:one-param}. Here all lines that define the surface cross at the critical point $(\tau,\tau,\tau)$, where $0 < \tau < 1$. We can derive a closed-form expression of $p(f^\tau)$ as a function of $\tau$, and optimize $\tau$. For the best possible $\tau$ this results in $p(f^\tau) < 0.24149$.

\begin{figure}
    \begin{subfigure}[t]{.47\linewidth}\centering
        \vspace{0pt}
        \includegraphics[width=\linewidth]{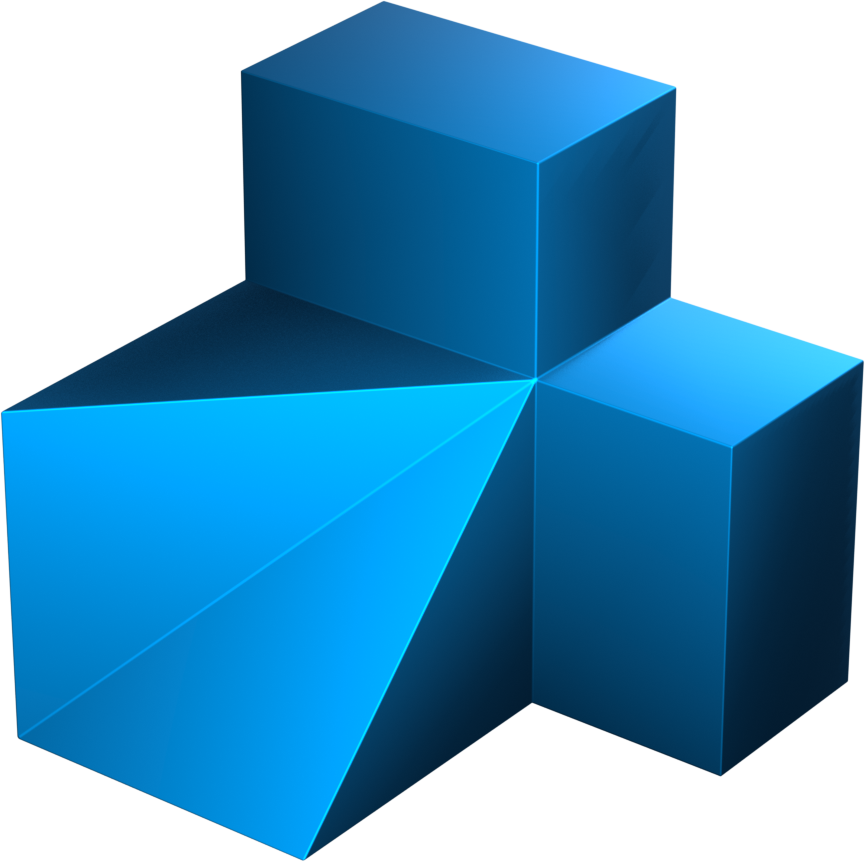}
        \caption{The algorithm $f^\tau$ that achieves $p(f^\tau) < 0.24149$. The blue area shows points in the $x,y,z$ space with $f(x,y,z) = 1$. All key lines that define the surface meet at the critical point $x=y=z=\tau$, and here we have optimized $\tau$ to minimize $p(f^\tau)$.}\label{fig:one-param}
    \end{subfigure}
    \hfill
    \begin{subfigure}[t]{.47\linewidth}
        \centering
        \vspace{0pt}
        \includegraphics[width=\linewidth]{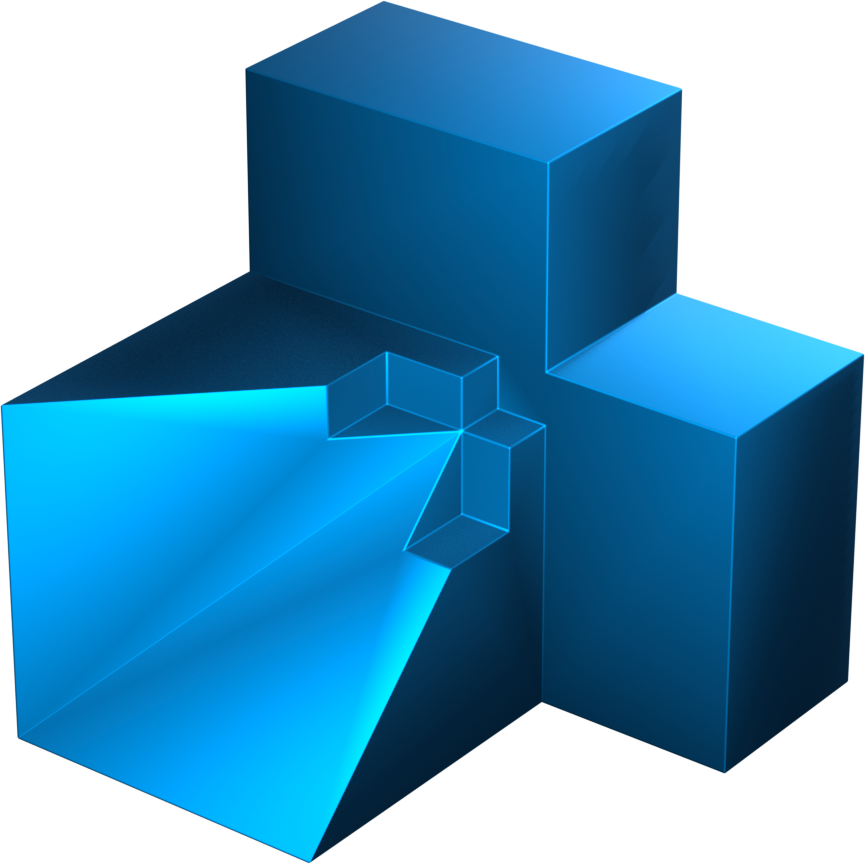}
        \caption{The three-parameter algorithm $f$ that achieves $p(f) < 0.24118$.}\label{fig:upper-bound}
    \end{subfigure}
    \caption{Illustrations of algorithms.}
\end{figure}

We then further tuned $f^\tau$ and discovered that modifications shown in \cref{fig:upper-bound} lead to a family of algorithms defined with three parameters that can slightly outperform $f^\tau$. Optimizing these parameters results in an algorithm $f$ with $p(f) < 0.24118$; again, see \cite{2-coloring-1-round-lean-nonanon} for the details.

\section*{Acknowledgments}

This work was supported in part by the Research Council of Finland, Grant 359104, and by the Quantum Doctoral Education Pilot, the Ministry of Education and Culture, decision VN/3137/2024-OKM-4.

\DeclareUrlCommand{\path}{}
\urlstyle{sf}
\bibliographystyle{alphaurl}
\bibliography{da.bib,extra.bib}
    
\end{document}

%% file: preamble.tex
\usepackage{amsthm}
\usepackage{graphicx}
\usepackage{amsmath, amssymb, amsfonts, verbatim}
\usepackage{hyphenat, subcaption, multirow}
\usepackage{nicefrac}
\usepackage{paralist}
\usepackage{microtype}
\usepackage{mathtools, etoolbox, enumitem}
\usepackage{thmtools}
\usepackage{thm-restate}
\usepackage{xparse}

\usepackage[margin=1in]{geometry}

\setitemize{noitemsep,topsep=0pt}
\setenumerate{noitemsep,topsep=0pt}

\usepackage{xcolor}
\definecolor{myblue}{HTML}{0088cc}
\usepackage[linktocpage=true,
	colorlinks,
	linkcolor=black,citecolor=black,urlcolor=myblue,
    bookmarks,bookmarksopen,bookmarksnumbered]
	{hyperref}
\usepackage[capitalise]{cleveref}

\newcommand{\fakeparagraph}[1]{\medskip\noindent\textbf{#1}}

\newtheorem{theorem}{Theorem}
\newtheorem{lemma}{Lemma}[section]
\newtheorem*{lemma*}{Lemma}

\newtheorem*{proposition*}{Proposition}

\newtheorem*{claim*}{Claim}

\theoremstyle{definition}

\newtheorem*{problem*}{Problem}

\theoremstyle{remark}

\newtheorem*{remark*}{Remark}

\crefname{lemma}{Lemma}{Lemmas}
\crefname{claim}{Claim}{Claims}
\crefname{enumi}{Step}{Steps}
\crefname{step}{Step}{Step}

\renewcommand{\leq}{\leqslant}

\renewcommand{\le}{\leqslant}
\renewcommand{\ge}{\geqslant}